\documentclass[11pt]{article}

\usepackage{fullpage}
\usepackage[dvips]{color}
\usepackage{latexsym}
\usepackage[activeacute,english]{babel}
\usepackage{graphicx}
\usepackage{amsmath}
\usepackage{algorithmic}
\usepackage{algorithm}
\usepackage{natbib}
\usepackage{amsthm,appendix,url}

\usepackage{multirow}

\usepackage{amscd}
\usepackage{amsfonts}
\usepackage[tableposition=top]{caption}
\usepackage{wrapfig}
\usepackage{ifthen}
\usepackage[utf8]{inputenc}

\newtheorem{prop}{Proposition}

\begin{document}
\title{Gaussian Process-Based Bayesian Nonparametric Inference of Population Trajectories from Gene Genealogies}
\author{Julia A. Palacios and Vladimir N. Minin*\\
Department of Statistics, University of Washington\\
\texttt{*corresponding author: vminin@uw.edu}}
\date{}
\maketitle

\begin{abstract}
Changes in population size influence genetic diversity of the population and, as a result, leave 
a signature of these changes in individual genomes in the population. We are interested in the inverse problem of reconstructing 
past population dynamics from genomic data. We start with a standard framework based on the 
coalescent, a stochastic process that generates genealogies connecting randomly sampled individuals
from the population of interest. These genealogies serve as a glue between the population demographic history and
genomic sequences. It turns out that only the times of genealogical lineage coalescences contain information
about population size dynamics. Viewing these coalescent times as a point process, estimating population
size trajectories is equivalent to estimating a conditional intensity of this point process. Therefore, our inverse problem is similar to estimating an inhomogeneous Poisson process intensity function. We demonstrate how recent advances in Gaussian process-based nonparametric inference for Poisson processes can be extended to Bayesian nonparametric estimation of population size dynamics under the coalescent. We compare our Gaussian process (GP) approach to one of the state of the art Gaussian Markov random field (GMRF) methods for estimating population trajectories. Using simulated data, we demonstrate that our method has better accuracy and precision. Next, we analyze two genealogies reconstructed from real sequences of hepatitis C and human Influenza A viruses. In both cases, we recover more believed aspects of the viral demographic histories than the GMRF approach. We also find that our GP method produces more reasonable uncertainty estimates than the GMRF method. \\
\end{abstract}

\section{Introduction}
Statistical inference in population genetics increasingly relies on the coalescent \citep{Kingman:1982uj}, the probability model that describes the relationship between a gene genealogy of a random sample of molecular sequences and  effective population size. This model provides a good approximation to the distribution of ancestral histories that arise from classical population genetics models \citep{rosenberg_genealogical_2002}. More importantly, coalescent-based inference methods allow us to estimate population genetic parameters, including population size trajectories, directly from genomic sequences \citep{griffiths_sampling_1994}. Recent examples of coalescent-based population dynamics estimation include reconstructing demographic histories of musk ox \citep{campos_2010} 
from fossil DNA samples and elucidating patterns of genetic diversity of the
dengue virus \citep{bennett_epidemic_2010}. 
\par
Here, we are interested in estimating effective population size trajectories from gene genealogies. The \textit{effective population size} is an abstract parameter that for a real biological population is proportional to the rate at which genetic diversity is lost or gained. In the absence of natural selection, the effective population size can be used to approximate census population size by knowing the generation time in calendar units (e.g. years) and the population variability in number of offspring \citep{Wakeley10112008}. The latter quantity might be difficult to know; however, sometimes it suffices to analyze an arbitrarily rescaled population size trajectory, assuming the variability in number of offspring remains constant. The effective population size is equal to the census population size in an idealized Wright-Fisher model. The Wright-Fisher model is a simple and established model of neutral reproduction in population genetics that assumes random mating and non-overlapping generations. 
For some RNA viruses, for example human influenza A virus, the effective population size rescaled by generation time (3 to 4 days) cannot be interpreted directly as the effective number of infections because of the presence of strong natural selection. However, one can always adopt a more cautious interpretation of the effective population size as a measure of relative genetic diversity \citep{rambaut_genomic_2008,Frost2010}.
\par
Coalescent-based methods for estimation of population size dynamics have evolved from stringent parametric assumptions, such as constant population size or exponential growth \citep{griffiths_sampling_1994, kuhner_maximum_1998, drummond_estimating_2002}, to more flexible nonparametric approaches that assume piecewise linear population trajectories \citep{strimmer_exploring_2001, opgen-rhein_inference_2005, drummond_bayesian_2005, heled_bayesian_2008, minin_smooth_2008}. The latter class of methods is more appropriate in the absence of prior knowledge about the underlying demographic dynamics, allowing researchers to infer shapes of population size trajectories rather than to impose parametric constraints on these shapes. These nonparametric methods, however, model population dynamics by imposing \textit{a priori} piecewise continuous functions which require regularization either by smoothing or by controlling the number of change points, also \textit{a priori}. The former regularization -- which works better in practice \citep{minin_smooth_2008} -- is inherently difficult because these piecewise continuous functions are defined on intervals of varying size. The piecewise nature of these methods creates further modeling problems if one wishes to incorporate covariates into the model or impose constraints on population size dynamics \citep{minin_smooth_2008}.  In this paper, we propose to solve these problems by bringing the coalescent-based estimation of population dynamics up to speed with modern Bayesian nonparametric methods. Making this leap in statistical methodology will allow us  to avoid artificial discretization of population trajectories, to perform regularization without making arbitrary scale choices, and, in the future, to extend our method into a multivariate setting.
\par
Our key insight stems from the fact that the coalescent with variable population size is an inhomogeneous continuous-time Markov chain \citep{TavareNotes} and, therefore, can be viewed as an inhomogeneous point process \citep{andersen_statistical_1995}. In fact, all current Bayesian nonparametric methods of estimation of population size dynamics resemble early Bayesian approaches to nonparametric estimation of the Poisson intensity function via piecewise continuous functions \citep{Arjas_2008}. Estimation of the intensity function of an inhomogeneous Poisson process is a mature field that evolved from maximum likelihood estimation under parametric assumptions \citep{brillinger_analyzing_1979} to frequentist \citep{diggle_kernel_1985} and, more recently, Bayesian nonparametric methods \citep{Arjas_2008, mller_log_1998, kottas_bayesian_2005, adams_tractable_2009}.
\par
Following \citet{adams_tractable_2009}, we \textit{a priori} assume that population trajectories follow a transformed Gaussian process (GP), allowing us to model the population trajectory as a continuous function. This is a convenient way to specify prior beliefs without a particular functional form on the population trajectory. The drawback of such a flexible prior is that the likelihood function involves integration over an infinite-dimensional random object and, as a result, likelihood evaluation becomes intractable. Fortunately, we are able to avoid this intractability and perform inference exactly by adopting recent algorithmic developments proposed by \citet{adams_tractable_2009}. 
We achieve tractability by a novel data augmentation for the coalescent process that relies on thinning algorithms for 
simulating the coalescent.
\par
Thinning is an accept/reject algorithm that was first proposed by \citet{lewis_simulation_1978} for the 
simulation of inhomogeneous Poisson processes and was later extended to a more general class of point processes by \citet{ogata2002lewis}. In the spirit of \citet{ogata2002lewis}, we develop novel thinning algorithms for the simulation of the coalescent. These algorithms, interesting in their own right, open the door for latent variable representation of the coalescent. This representation leads to a new data augmentation that is computationally tractable and amenable to standard Markov chain Monte Carlo (MCMC) sampling from the posterior distribution of model parameters and latent variables.
\par
We test our method on simulated data and compare its performance with a representative piecewise linear approach, a Gaussian Markov random field (GMRF) based method \citep{minin_smooth_2008}. We demonstrate that our method is more accurate and more precise than the GMRF method in all simulation scenarios. We also apply our method to two real data sets that have been previously analyzed in the literature: a hepatitis C virus (HCV) genealogy estimated from sequences sampled in 1993 in Egypt and a genealogy of the H3N2 human influenza A virus estimated from sequences sampled in New York state between 2002 and 2005. In the HCV analysis, we successfully recover all believed key aspects of the population size trajectory. Compared to the GMRF method, our GP method better reflects the uncertainty inherent in the HCV data. In our second real data example, our GP method successfully reconstructs a population trajectory of the human influenza A virus with an expected seasonal series of peaks followed by population bottlenecks, while the GMRF method's reconstructed trajectory fails to recover some of the peaks and bottlenecks.

\section{Methods}
\subsection{Coalescent Background}
The coalescent model allows us to trace the ancestry of a random sample of $n$ genomic sequences. These ancestral relationships are represented by a genealogy or tree; the times at which two sequences or lineages merge into a common ancestor are called coalescent times. The coalescent with variable population size can be viewed as a non-homogeneous Markov death process that starts with $n$ lineages at present time $t_{n}=0$ and decreases by one, with time running backwards, until reaching one lineage at $t_{1}$, at which point the samples have been traced to their most recent common ancestor \citep{griffiths_sampling_1994}.
\begin{figure}
  \begin{center}
\includegraphics[scale=.7]{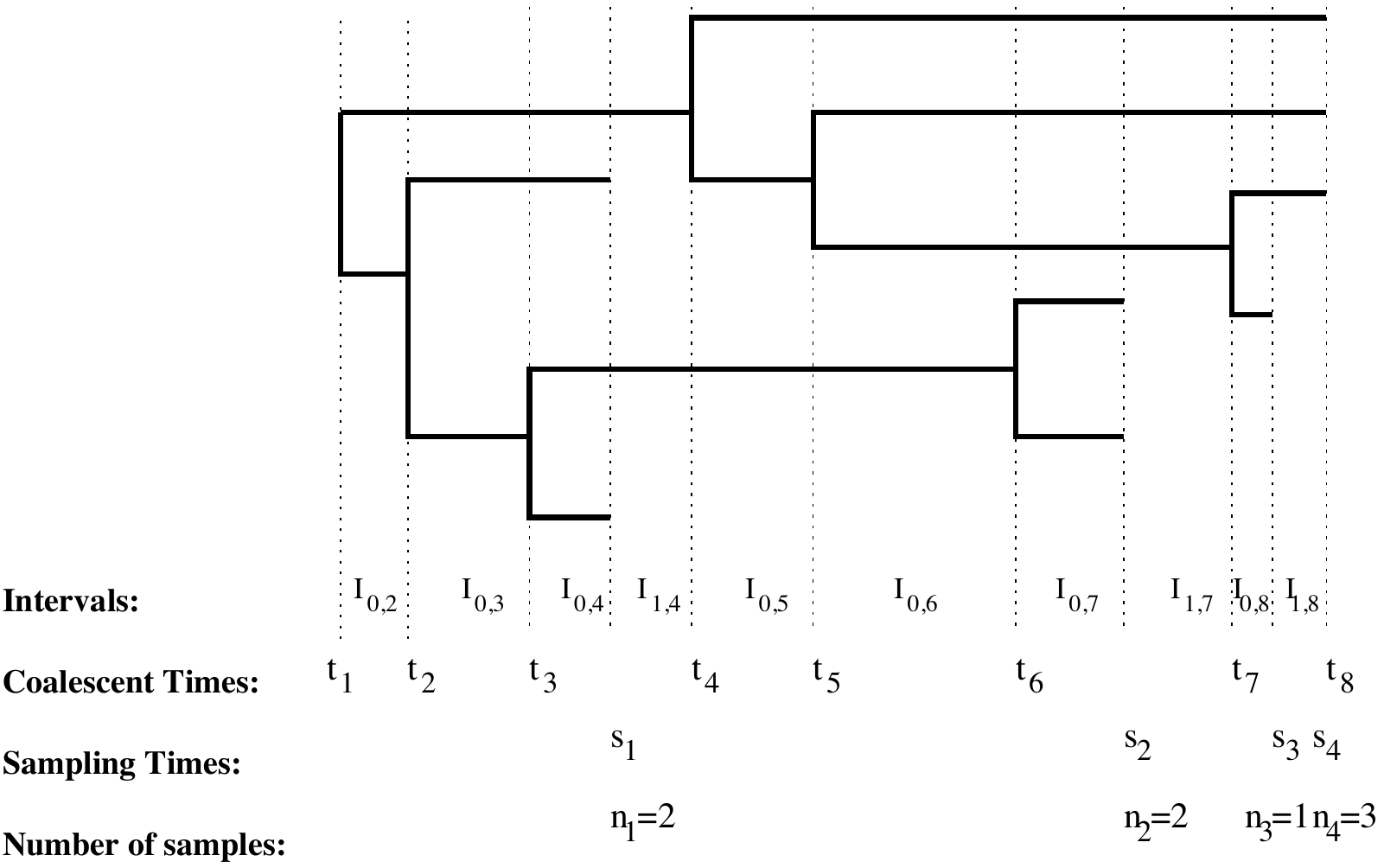}
  \end{center}
\caption{\small{Example of a genealogy relating serially sampled sequences (heterochronous sampling). The number of lineages changes every time we move between intervals ($I_{i,k}$). Each endpoint of an interval is a coalescent time ($\{t_{k}\}^{n}_{k=1})$ or a sampling time ($\{s_{j}\}^{m}_{j=1}$). The number of sequences sampled at time $s_{j}$ is denoted by $n_{j}$. }}
\end{figure}
Here, we assume that a genealogy with time measured in units of generations is observed. The shape of the genealogy depends on the effective population size trajectory, $N_{e}(t)$, and the number of samples accumulated through time: the larger the effective population size, the longer two lineages need to wait to meet a common ancestor and the larger the sample size, the faster two lineages coalesce.  Formally, let $t_{n}=0$ denote the present time when all the $n$ available sequences are sampled (\textit{isochronous coalescent}) and let $t_{n}=0<t_{n-1}<\cdots<t_{1}$ denote the coalescent times of lineages in the genealogy with time going backwards. Then, the conditional density of the 
coalescent time $t_{k-1}$ takes the following form: 
\begin{equation}
P[t_{k-1}|t_{k},N_{e}(t)]=\frac{C_{k}}{N_{e}(t_{k-1})} \exp \left\lbrace  -\int^{t_{k-1}}_{t_{k}}\frac{C_{k}}{N_{e}(t)}dt \right\rbrace  , 
\end{equation}
 where $C_{k}=\binom{k}{2}$ is the coalescent factor that depends on the number of lineages $k=2,\ldots,n$.
 \par
The \textit{heterochronous coalescent} arises when samples of sequences are collected at different times (Figure 1). Such serially sampled data are common in studies of rapidly evolving viruses and analyses of ancient DNA \citep{campos_2010}. Let $t_{n}=0<t_{n-1}<\cdots<t_{1}$ denote the coalescent times as before, but now  let $s_{m}=0<s_{m-1}<\cdots<s_{1}<s_{0}=t_{1}$ denote  sampling times of $n_{m},\ldots,n_{1}$ sequences respectively, $\sum^{m}_{j=1}n_{j}=n$. Further, let $\mathbf{s}$ and $\mathbf{n}$ denote the vectors of sampling times and numbers of sequences sampled at these times, respectively (Figure 1). Now, the coalescent factor changes not only at the coalescent events but also at the sampling times. Let 
\begin{equation}
I_{0,k}=(max\{t_{k},s_{j}\},t_{k-1}] \text{, for } s_{j}<t_{k-1} \text{ and }k=2,\dots,n,
\end{equation}
be the intervals that end with a coalescent event and
\begin{equation}
I_{i,k}=(max\{t_{k},s_{j+i}\},s_{j+i-1}] \text{, for } s_{j+i-1}>t_{k} \text{ and } s_{j}<t_{k-1}, k=2,\ldots,n,
\end{equation}
be the intervals that end with a sampling event. We denote the number of lineages in $I_{i,k}$ with $n_{i,k}$. 
Then, for $k=2,...,n$,
\begin{equation}
P[t_{k-1}|t_{k},\mathbf{s},\mathbf{n},N_{e}(t)]=\frac{C_{0,k}}{N_{e}(t_{k-1})} \exp -\left\lbrace\int_{I_{0,k}} \frac{C_{0,k}}{N_{e}(t)}dt+\sum^{m}_{i=1}\int_{I_{i,k}} \frac{C_{i,k}}{N_{e}(t)}dt\right\rbrace,
\end{equation}
where $C_{i,k}=\binom{n_{i,k}}{2}$. That is, the density for the next coalescent time $t_{k-1}$ is the product of the density of the coalescent time $t_{k-1} \in I_{0,k}$ and the probability of not having a coalescent event during the period of time spanned by intervals $I_{1,k},\ldots,I_{m,k}$ \citep{joseph_coalescent_1999}.
\subsection{Gaussian Process Prior for Population Size Trajectories}
For both isochronous or heterochronous data, we place the same prior on $N_{e}(t)$:
\begin{equation}
N_{e}(t)=\left[ \frac{\lambda}{1+\exp\left\lbrace -f(t)\right\rbrace }\right]^{-1},
\end{equation}
where
\begin{equation}
f(t)\sim \mathcal{GP}(\mathbf{0},\mathbf{C}(\mbox{\boldmath$\theta$}))
\end{equation}
and $\mathcal{GP}(\mathbf{0},\mathbf{C}(\mbox{\boldmath$\theta$}))$ denotes a Gaussian process with mean function $\mathbf{0}$ and covariance function $\mathbf{C}(\mbox{\boldmath$\theta$})$ with hyperparameters {\boldmath$\theta$}. \textit{A priori}, $1/N_{e}(t)$ is a Sigmoidal Gaussian Process,  a scaled logistic function of a Gaussian process whose range is restricted to lie in $[0,\lambda]$; $\lambda$ is a positive constant hyperparameter, the inverse of which serves as a lower bound of $N_{e}(t)$ \citep{adams_tractable_2009}.
\par
A \textit{Gaussian process} is a stochastic process such that any finite sample from the process has a joint Gaussian distribution. The process is completely specified by its mean and covariance functions \citep{rasmussen_gaussian_2005}. For computational convenience we use Brownian motion as our Gaussian process prior.
Generating a finite sample from a Gaussian processes requires $\mathcal{O}(n^{3})$ computations due to the inversion of the covariance matrix. However, when the precision matrix, the inverse of the covariance, is sparse, such simulations can be accomplished much faster \citep{rue_gaussian_2005}. For example, when we choose to work with a Brownian motion with covariance matrix elements $C(t,t^{'})=\frac{1}{\theta}(min(t, t^{'}))$ and precision parameter $\theta$, then the inverse of this matrix is tri-diagonal, which reduces the computational complexity of simulations from $\mathcal{O}(n^{3})$ to $\mathcal{O}(n)$. In our MCMC algorithm, we need to generate realizations from the Gaussian processes at thousands of points, so the speed-up afforded by the Brownian motion becomes almost a necessity, prompting us to use this process as our prior in all our examples.
\subsection{Priors for Hyperparameters}
The precision parameter $\theta$ controls the degree of autocorrelation of our Brownian motion prior and influences the ``smoothness" of the reconstructed population size trajectories. We place on $\theta$ a Gamma prior distribution with parameters $\alpha$ and $\beta$. The other hyperparameter in our model is the upper bound of $1/N_{e}(t)$, $\lambda$. When this upper bound $\lambda$ 
is unknown, the model is unidentifiable (see Equation (5)). However, in many circumstances it is possible 
to obtain an upper bound $\lambda$ (or equivalently, a lower bound on $N(t)$) from previous studies and use this value to define 
the prior distribution of $\lambda$. We use the following strategy to construct an informative prior for $\lambda$. Let $\hat{\lambda}$ denote our best guess of the upper bound, possibly obtained from previous studies. Then, the prior on $\lambda$ is a mixture of 
a uniform distribution for values to the left of $\hat{\lambda}$ and an exponential distribution to the right:
\begin{equation}
P(\lambda) =  \epsilon \frac{1}{\hat{\lambda}} I_{\left\lbrace \lambda<\hat{\lambda}\right\rbrace }+ (1-\epsilon) \frac{1}{\hat{\lambda}}e^{ -\frac{1}{\hat{\lambda}}(\lambda-\hat{\lambda})}I_{\left\lbrace \lambda \geq \hat{\lambda}\right\rbrace },
\end{equation} 
where $\epsilon>0$ is a mixing proportion. When $\hat{\lambda}$ is considerably smaller than the unknown $\lambda$, the recovered curve will be visibly truncated around $\hat{\lambda}$, indicating that one needs to try higher values of $\hat{\lambda}$.
\subsection{Doubly Intractable Posterior}
Coalescent times $\mathcal{T}=\{t_{n},t_{n-1},\ldots,t_{1}\}$ of a given genealogy contain information needed to estimate $N_{e}(t)$ (see Equations (1) and (4)). Given that $N_{e}(t)$ is a one-to-one function of $f(t)$ (Equation (5)), we will focus the discussion on the inference of $f(t)$. The posterior distribution of $f(t)$ and hyperparameters $\theta$ and $\lambda$ becomes
\begin{equation}
P(f(t),\theta,\lambda|\mathcal{T}) \propto P(\mathcal{T}|\lambda,f(t))P(f(t)|\theta)P(\theta)P(\lambda) ,
\end{equation}
where $P(f(t)|\theta)$ is a Gaussian process prior with hyperparameter $\theta$ and
\begin{equation}
P(\mathcal{T}|\lambda,f(t) )=\prod^{n}_{k=2}\frac{C_{k}\lambda}{1+\exp\left\lbrace -f(t_{k-1})\right\rbrace }\exp \left[   -C_{k}\int^{t_{k-1}}_{t_{k}}\frac{\lambda}{1+\exp\left\lbrace -f(t)\right\rbrace }dt\right] 
 \end{equation}
is the likelihood function for the isochronous data (heterochronous data likelihood has a similar form). The integral in the exponent of Equation (9) and the normalizing constant of Equation (8) are computationally intractable, making the posterior doubly intractable \citep{murray_mcmc_2006}.
\par
\citet{adams_tractable_2009} faced a similar doubly intractable posterior distribution in the context of nonparametric estimation of intensity of the inhomogeneous Poisson process. These authors propose an introduction of latent variables so that the augmented data likelihood becomes tractable. This tractability makes the posterior distribution of latent variables and model parameters amenable to standard MCMC algorithms. Since \citet{adams_tractable_2009} based their data augmentation on the thinning algorithm for simulating inhomogeneous Poisson processes, we would like to devise a similar data augmentation based on a thinning algorithm for simulation of the coalescent with variable population size. In this simulation, we envision generating coalescent times assuming a constant population size and then thinning these times so that the distribution of the remaining (non-rejected) coalescent times follows the coalescent with variable population size. Since no thinning algorithm for simulating the coalescent process exists, we develop a series of such algorithms. In developing these algorithms, we find it useful to view the coalescent as a point process, a representation that we discuss below.
\subsection{The Coalescent as a Point Process}
The joint density of coalescent times is obtained by multiplying the conditional densities defined in Equations (1) or (4). This density can be expressed as
\begin{equation}
P(t_{1},\ldots,t_{n-1}|N_{e}(t))=\prod^{n}_{k=2}\lambda^{*}(t_{k-1}|t_{k})\exp\left\lbrace -\int^{t_{k-1}}_{t_{k}}\lambda^{*}(t|t_{k})dt\right\rbrace ,
\end{equation}
where $\lambda^{*}(t|t_{k})$ denotes the conditional intensity function of a point process on the real line \citep{daley_introduction_2002}. For isochronous coalescent, the conditional intensity is defined by the step function:
\begin{equation}
\lambda^{*}(t|t_{k})=\binom{k}{2}N_{e}(t)^{-1}  1_{\{ t \in (t_{k}, t_{k-1}]\}}, \text{ for } k=2,\ldots,n,
\end{equation}
and the conditional intensity of the heterochronous coalescent point process is:
\begin{equation}
\lambda^{*}(t|\mathbf{n},\mathbf{s},t_{k})=\sum^{m}_{i=1}\binom{n_{i,k}}{2}N_{e}(t)^{-1} 1_{\{ t \in I_{i,k}\} }, \text{ for } k=2,\ldots,n.
\end{equation}
This novel representation allows us to reduce the task of estimating $N_{e}(t)$ to the estimation of the inhomogeneous intensity of the coalescent point process and to develop simulation algorithms based on thinning. 
\subsection{Coalescent Simulation via Thinning}
To the best of our knowledge, the only method available for simulating the coalescent under the deterministic variable population size model is a time transformation method \citep{slatkin_pairwise_1991, hein_gene_2005}. This method is based on the random time-change theorem due to \citet{papangelouo_integrability_1972}. 
Under the time transformation method, to simulate coalescent times, we proceed sequentially starting with $k=n$ and $t_{n}=0$, generating $t$ from an exponential distribution with unit mean, solving
\begin{equation}
t=\int^{t_{k-1}}_{t_{k}}\lambda^{*}(u|t_{k})du
\end{equation} 
for $t_{k-1}$ analytically or numerically and repeating the procedure until $k=2$. For isochronous coalescent,  $\lambda^{*}(u|t_{k})$ is defined in Equation (11) and for the heterochronous coalescent, $\lambda^{*}(u|t_{k})=\lambda^{*}(u|\mathbf{n},\mathbf{s},t_{k})$ is the piecewise function defined in Equation (12). When $N_{e}(t)$ is stochastic, the integral in Equation (13) becomes intractable and the time transformation method is no longer practical. Instead, we propose to use \textit{thinning}, a rejection-based method that does not require calculation of the integral in Equation (13).
\par
\citet{lewis_simulation_1978} proposed thinning a homogeneous Poisson process for the simulation of an inhomogeneous Poisson process with intensity $\lambda(t)$. The idea is to start with a realization of points from a homogeneous Poisson process with intensity $\lambda$ and accept/reject each point with acceptance probability $\lambda(t)/\lambda$, where $\lambda(t) \leq \lambda$. The collection of accepted points forms a realization of the inhomogeneous Poisson process with conditional intensity $\lambda(t)$. \citet{ogata2002lewis} extended Lewis and Shedler's thinning for the simulation of any point process that is absolutely continuous with respect to the standard Poisson process. We develop a series of thinning algorithms for the coalescent process that are similar to Ogata's algorithms, but not identical to them. Algorithm 1 outlines the simulation of $n$ coalescent times under the isochronous sampling. Given $t_{k}$, we start generating and accumulating exponential random numbers $E_{i}$ with rate $C_{k}\lambda$, until $t_{k-1}=t_{k}+E_{1}+E_{2}+\ldots$ is accepted with probability $1/N_{e}(t_{k-1})\lambda$. (see Web Appendix A for details and simulation algorithms of coalescent times for heterochronous sampling). In order to ensure convergence of the algorithm, we require $\int^{\infty}_{0} \frac{du}{N_{e}(u)} = \infty$ a.s., which is equivalent to requiring that all sampled lineages can be traced back to their single common ancestor with probability 1. Notice that $N_{e}(t)$ can be either deterministic or stochastic. The latter case is considered in Web Supplementary Algorithm 2, where we work with $f(t)$ instead of $N_{e}(t)$ for notational convenience. 
\par
If $N_{e}(t)$ is deterministic and equation (13) can be solved analytically, the time transformation method is likely to be more efficient that thinning since the thinning algorithm is  an accept/reject algorithm with the acceptance probability highly dependent on the definition of $\lambda$. However, efficiency of the thinning algorithm can be improved by replacing the constant upper bound $\lambda$ on $1/N_{e}(t)$, by a piece-wise constant or a piece-wise linear function of local upper bounds in order to achieve higher acceptance probabilities, similarly to the adaptive rejection sampling of \citet{Gilks1992}.
\renewcommand{\algorithmicrequire}{\textbf{Input:}} 
\renewcommand{\algorithmicensure}{\textbf{Output:}}
\linespread{0.9}
\begin{algorithm}
\caption{Simulation of isochronous coalescent times by thinning - $N_{e}(t)$ is a deterministic function} 
\begin{algorithmic} [1]
\REQUIRE {$k=n$, $t_{n}=0$, $t=0$, $1/N_{e}(t) \leq \lambda$, $N_{e}(t)$}
\ENSURE {$\mathcal{T}=\{t_{k}\}^{n}_{k=1}$}
\WHILE {$k>1$}
\STATE Sample $E\sim Exponential (C_{k}\lambda)$ and $U\sim U(0,1)$
\STATE t=t+E
\IF{$U \leq \frac{1}{N_{e}(t)\lambda} $}
\STATE $k \leftarrow k-1$,  $t_{k} \leftarrow t$
\ENDIF
\ENDWHILE
\end{algorithmic} 
\end{algorithm}
\renewcommand{\baselinestretch}{1.8}
\vspace{-0.4cm}
\subsection{Data Augmentation and Inference}
As mentioned in the previous section, our thinning algorithm for the coalescent is motivated by our desire to construct a data augmentation scheme. We imagine that observed coalescent times $\mathcal{T}$ were generated by the thinning procedure described in Algorithm 1, so we augment $\mathcal{T}$ with rejected (thinned) points $\mathcal{N}$. If we keep track of the rejected points resulting from Algorithm 1, then, given $t_{k}$, $f(t_{k}), \mathbf{f}_{\mathcal{N}_{k}}=\left\lbrace f(t_{k,i})\right\rbrace^{m_{k}}_{i=1}$ and $\lambda$, we start proposing new time points $\mathcal{N}_{k}=\{t_{k,1},\ldots,t_{k,m_{k}}\}$ until $t_{k-1}$ is accepted, so that
\begin{eqnarray}
P(t_{k-1},\mathcal{N}_{k}|t_{k},f(t_{k-1}), \mathbf{f}_{\mathcal{N}_{k}},\lambda)=(C_{k}\lambda) ^{m_{k}+1}\exp \left\lbrace -C_{k} \lambda (t_{k}-t_{k-1})\right\rbrace \left[\frac{1}{1+\exp\left\lbrace -f(t_{k-1})\right\rbrace}\right]  \notag \\ \times\prod^{m_{k}}_{i=1} \left[1-\frac{1}{1+\exp\left\lbrace -f(t_{k,i})\right\rbrace }\right].
\end{eqnarray}
For the heterochronous coalescent (see Algorithm 3 in Web Supplemental Materials), Equation (14) is modified in the following way: 
\begin{eqnarray}
P(t_{k-1},\mathcal{N}_{k}|t_{k},f(t_{k-1}), \mathbf{f}_{\mathcal{N}_{k}},\lambda,\mathbf{s},\mathbf{n})=
(\lambda C_{0,k})^{1+m_{0,k}} \exp\{-\lambda C_{0,k}l(I_{0,k})\} \notag\\
\times \left[ \left( \frac{1}{1+\exp\left\lbrace  -f(t_{k-1})\right\rbrace  } \right) \prod^{m_{k}}_{i=1}\frac{1}{1+\exp\left\lbrace f(t_{k,i})\right\rbrace } \right]\prod^{m}_{i=1} \left[  (\lambda C_{i,k})^{m_{i,k}} \exp\{-\lambda C_{i,k}l(I_{i,k})\}\right],
\end{eqnarray}
where $l(I_{i,k})$ denotes the length of the interval $I_{i,k}$ and $m_{i,k}=\sum^{m_{k}}_{l=1} 1 \left\lbrace  t_{k,l} \in I_{i,k} \right\rbrace$ denotes the number of latent points in interval $I_{i,k}$ .
Let $\mathbf{f}_{\mathcal{T},\mathcal{N}}=\left\{\left\lbrace f(t_{k}) \right\rbrace^{n}_{k=1}, \left\lbrace\left\lbrace f(t_{k,i}) \right\rbrace^{m_{k}}_{i=1}\right\rbrace^{n}_{k=2}\right\}$, then the augmented data likelihood of $\mathcal{T}$ and $\mathcal{N}$ becomes
\begin{equation}
P(\mathcal{T},\mathcal{N}|\mathbf{f}_{\mathcal{T},\mathcal{N}},\lambda)=\prod^{n}_{k=2}P(t_{k-1},\mathcal{N}_{k}|t_{k},f(t_{k-1}), \mathbf{f}_{\mathcal{N}_{k}},\lambda).
\end{equation}
Then, the posterior distribution of $f(t)$ and hyperparameters evaluated at the observed $\mathcal{T}$ and latent $\mathcal{N}$ time points is
\begin{equation}
P(\mathbf{f}_{\mathcal{T},\mathcal{N}},\lambda,\theta|\mathcal{T},\mathcal{N}) \propto P(\mathcal{T},\mathcal{N}|\mathbf{f}_{\mathcal{T},\mathcal{N}},\lambda) P(\mathbf{f}_{\mathcal{T},\mathcal{N}}|\theta)P(\lambda)P(\theta).
\end{equation}
The augmented posterior can now be easily evaluated since it does not involve integration of infinite-dimensional random functions.
We follow \cite{adams_tractable_2009} and develop a MCMC algorithm to sample from the posterior distribution (16).  
At each iteration of our MCMC, we update the following variables:  (1) number of ``rejected" points $\#\mathcal{N}$; (2) the locations of the rejected points $\mathcal{N}$; (3) the function values $\mathbf{f}_{\mathcal{T},\mathcal{N}}$ and (4) the hyperparameters $\theta$ and $\lambda$. We use a Metropolis-Hastings algorithm to sample the number and locations of latent points and the hyperparameter $\lambda$; we Gibbs sample the hyperparameter $\theta$. Updating the function values $\mathbf{f}_{\mathcal{T},\mathcal{N}}$ is nontrivial, because this high dimensional vector has correlated components. In such cases, single-site updating is inefficient and block updating is preferred \citep {rue_gaussian_2005}. We use elliptical slice sampling, proposed by \citet{murray_elliptical_2010}, to sample $\mathbf{f}_{\mathcal{T},\mathcal{N}}$. The advantage of using the elliptical slice sampling proposal is that it does not require the specification of tunning parameters and works well in high dimensions. The details of our MCMC algorithm can be found in Web Appendix B. 
\par
We summarize the posterior distribution of $N_{e}(t)$ by its empirical median and $95\%$ Bayesian credible intervals (BCIs) evaluated at a grid of points. This grid can be made as fine as necessary after the MCMC is finished. That is, given the function values $\mathbf{f}_{\mathcal{T},\mathcal{N}}$ at coalescent and latent time points, and the value of the precision parameter $\theta$ at each iteration,  we sample the function values at a grid of points $\mathbf{g}=\left\lbrace g_{1},...,g_{B}\right\rbrace$ from its predictive distribution $\mathbf{f}_{\mathbf{g}} \sim P(\mathbf{f}_{\mathbf{g}}|\mathbf{f}_{\mathcal{T},\mathcal{N}},\theta)$, and evaluate $\left\lbrace N_{e}(g_{i}) \right\rbrace^{B}_{i=1}$.
\section{Results}
\subsection{Simulated Data}
We simulate three genealogies relating $100$ individuals, sampled at the same time $t=0$ (isochronous sampling) under the following demographic scenarios:
1) constant population size trajectory: $N_{e}(t)=1$; 
2) exponential growth: $N_{e}(t)=25e^{-5t}$; and
3) population expansion followed by a crash: $N_{e}(t)=e^{4t} 1_{\left\lbrace  t \in [0,0.5]\right\rbrace }+ e^{-2t+3} 1_{\left\lbrace t \in (0.5,\infty)\right\rbrace }.$
We compare the posterior median with the truth by the sum of relative errors (SRE):
\begin{equation}
SRE=  \sum^{K}_{i=1}\frac{|\hat{N}_{e}(s_{i})-N_{e}(s_{i})|}{N_{e}(s_{i})},
\end{equation}
where $\hat{N}_{e}(s_{i})$ is the estimated trajectory at time $s_{i}$ with $s_{1}=t_{1}$, the time to the most recent common ancestor and $s_{K}=t_{n}=0$ for any finite $K$. Similarly, we compute the mean relative width (MRW) of the $95\%$ BCIs defined in the following way:
\begin{equation}
MRW=  \sum^{K}_{i=1}\frac{|\hat{N}_{97.5}(s_{i})-\hat{N}_{2.5}(s_{i})|}{K N_{e}(s_{i})}.
\end{equation}
We also compute the percentage of time, the $95\%$ BCIs cover the truth (envelope) in the following way:
\begin{equation}
 envelope=\frac{\sum^{K}_{i=1}I(\hat{N}_{2.5}(s_{i})\leq N(s_{i}) \leq \hat{N}_{97.5}(s_{i}))}{K}.
\end{equation}
As a measure of the frequentist coverage, we calculate the percentage of times the truth is completely covered by the $95\%$ BCIs (envelope = 1), by simulating each demographic scenario and performing Bayesian estimation of each such simulation $100$ times.
\par
We compute the three statistics for the three simulation scenarios for $K=150$ at equally spaced time points (Table 1). These statistics do not change significantly when we use higher values of $K$.
Additionally, we compute the variation of $\hat{N}_{e}(t)$ over a regular grid of $K=150$ points as follows: 
\begin{equation}
variation=\sum^{K-1}_{i=1}|\hat{N}_{e}(s_{i+1})-\hat{N}_{e}(s_{i})|,
\end{equation}
\par
For all simulations, we set the mixing parameter $\epsilon$ of the prior density for $\lambda$ (Equation (7))  to $\epsilon=0.01$. The parameters of the Gamma prior on the GP precision parameter $\theta$ were set to $\alpha=\beta=0.001$. We summarize our posterior inference in Figure 2 and compare our GP method to the GMRF smoothing method \citep{minin_smooth_2008}. The effective population trajectory is log transformed and time is measured in units of generations.
\par 
For the constant population scenario (first row in Figure 2),  the truth (dashed lines) is almost perfectly recovered by the GP method (solid black line) and the $95\%$ BCIs shown as gray shaded areas are remarkably tight. For the exponential growth simulation (second row), the GMRF method recovers the truth better in the right tail, while our GP method recovers it much better in the left tail. The higher variation of the GP reconstruction in the right tail makes this measure higher than for the GMRF reconstruction. Overall, our GP method better recovers the truth in the exponential growth scenario, as evidenced by SREs and MRWs in Table 1. The last row in Figure 2 shows the results for a population that experiences expansion followed by a crash in effective population size. In this case, $95\%$ BCIs of the two methods do not completely cover the true trajectory. While an area near the bottleneck is particularly problematic, the GP method's envelope is much higher ($92\%$) than the envelope produced by the GMRF method ($77.3\%$), the variation recovered by the GP method is closer to the true variation in all simulation scenarios and in general, in terms of the four statistics employed here, the GP method shows better performance. Results for the GMRF method were obtained using the \texttt{BEAST} software \citep{latest_beast} with running times ranging from 25 to 40 minutes, while results for the GP method were obtained using \texttt{R} with running times ranging from 60 to 180 minutes. Although our GP implementation takes longer, we obtain better performance in a still reasonable amount of time.
\linespread{0.9}
\begin{figure}
\vspace{-25pt}
  \begin{center}
 \includegraphics[scale=.95]{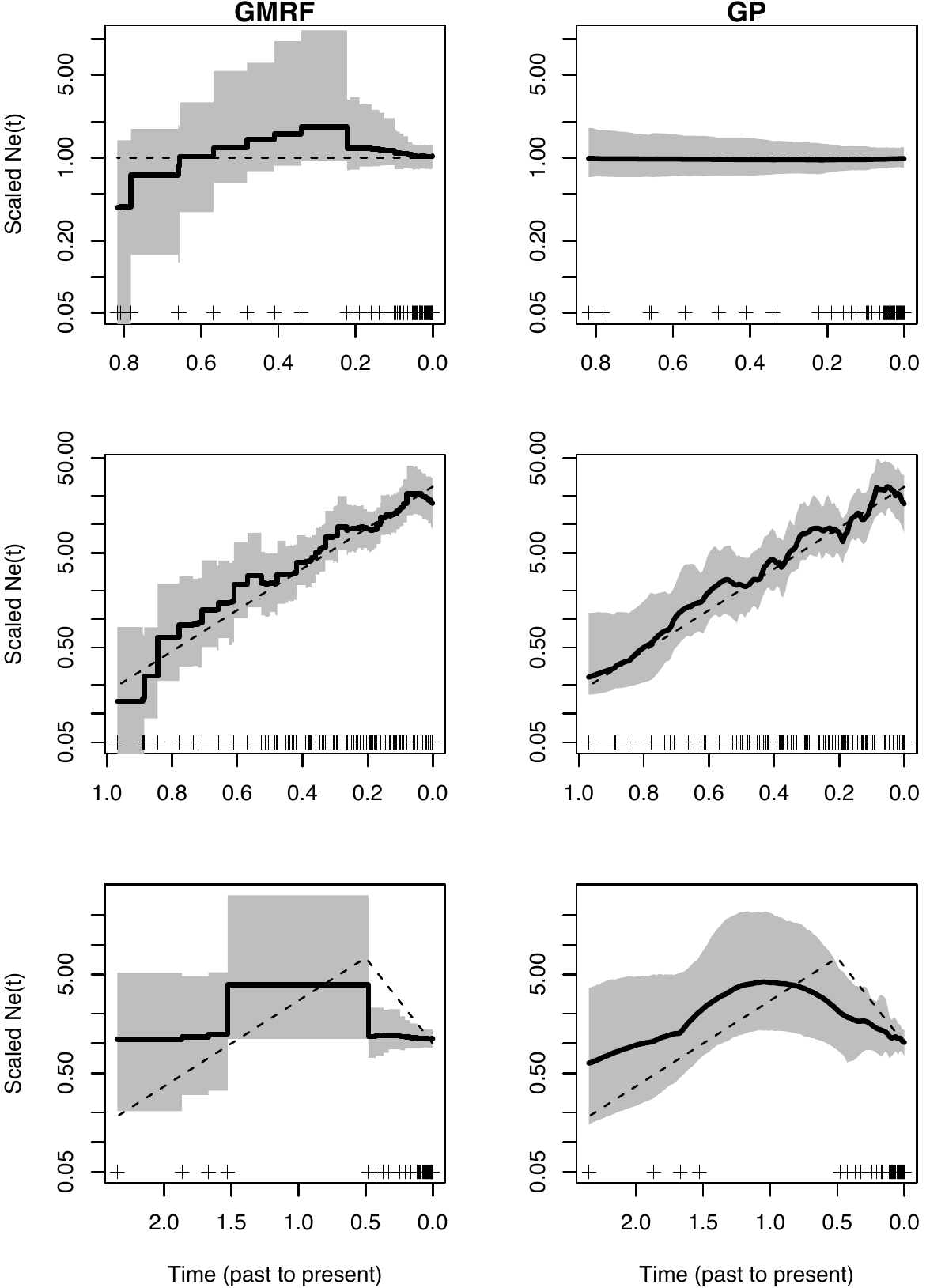}
  \end{center}
  \vspace{-10pt}
  \caption{\small{Simulated data under the constant population size (first row), exponential growth (second row) and expansion followed by a crash (third row). The simulated points are represented by the points at the bottom of each plot. We show the log of the effective population size trajectory estimated under the Gaussian Markov random field smoothing (GMRF) method and our method: Gaussian process-based nonparametric inference of effective population size (GP). We show the true trajectories as dashed lines, posterior medians as solid black lines and $95\%$ BCIs by gray shaded areas. }}
\end{figure}
\linespread{1.0}
\begin{table}
\caption{\small{Summary of Simulation Results Depicted in Figure 2. SRE is the sum of relative errors as defined in Equation (19), MRW is the mean relative width of the 95\% BCI as defined in Equation (20), envelope is calculated as in Equation (21) and variation is calculated as in Equation (22). }}
\centering \tabcolsep=0.05cm  \begin{tabular}{l c c c c c c c c c c c c c} 
\hline\hline
Simulations & \multicolumn{4}{c}{SRE} & \multicolumn{3}{c}{MRW} &  \multicolumn{2}{c}{Envelope} &  \multicolumn{4}{c}{Variation}\\
\cline{3-4} \cline{6-7} \cline{9-10} \cline{12-14}
& & GMRF & GP & & GMRF & GP & & GMRF & GP & & GMRF & GP & TRUTH\\
\hline
Constant &  &50.41 &\textbf{4.15}& &4.21&\textbf{0.72} & & \textbf{100.0\%} &\textbf{100.0\%} & & 2.27 &\textbf{0.08} & 0.00\\
Exp. growth & &47.65&\textbf{33.60}& &2.55&\textbf{2.35} & &\textbf{100.0\%} &\textbf{100.0\%} & & \textbf{30.19} &52.41 & 24.80\\
Expansion/crash & &181.88&\textbf{140.88} & &10.7&\textbf{7.26}& & 77.33\% &\textbf{92.0\%} & & 5.69 &\textbf{6.94} & 13.46\\[1ex] 
\hline 
\end{tabular}
\end{table}
\par
Next, we simulate each of the three scenarios $100$ times and compute \textbf{the four} statistics described before for both methods. The distributions of these statistics are represented by the boxplots depicted in Figure 3. In general, our GP method has smaller SREs, except in the constant case, where the distributions look very similar; smaller MRWs, larger envelopes and variation closer to the truth. Additionally, we calculate the percentage of times, the envelope is 1 as a proxy for frequentist coverage of the 95\% BCIs. Since the 95\% BCIs are calculated pointwise at $150$ equally spaced points,  we do not necessarily expect frequentist coverage to be close to 95\%. The results are shown as the numbers at the top of the right plot in Figure 3. The coverage levels obtained using the GP method are larger than those obtained using the GMRF method.
\begin{figure}
\vspace{-25pt}
  \begin{center}
 \includegraphics[scale=.6]{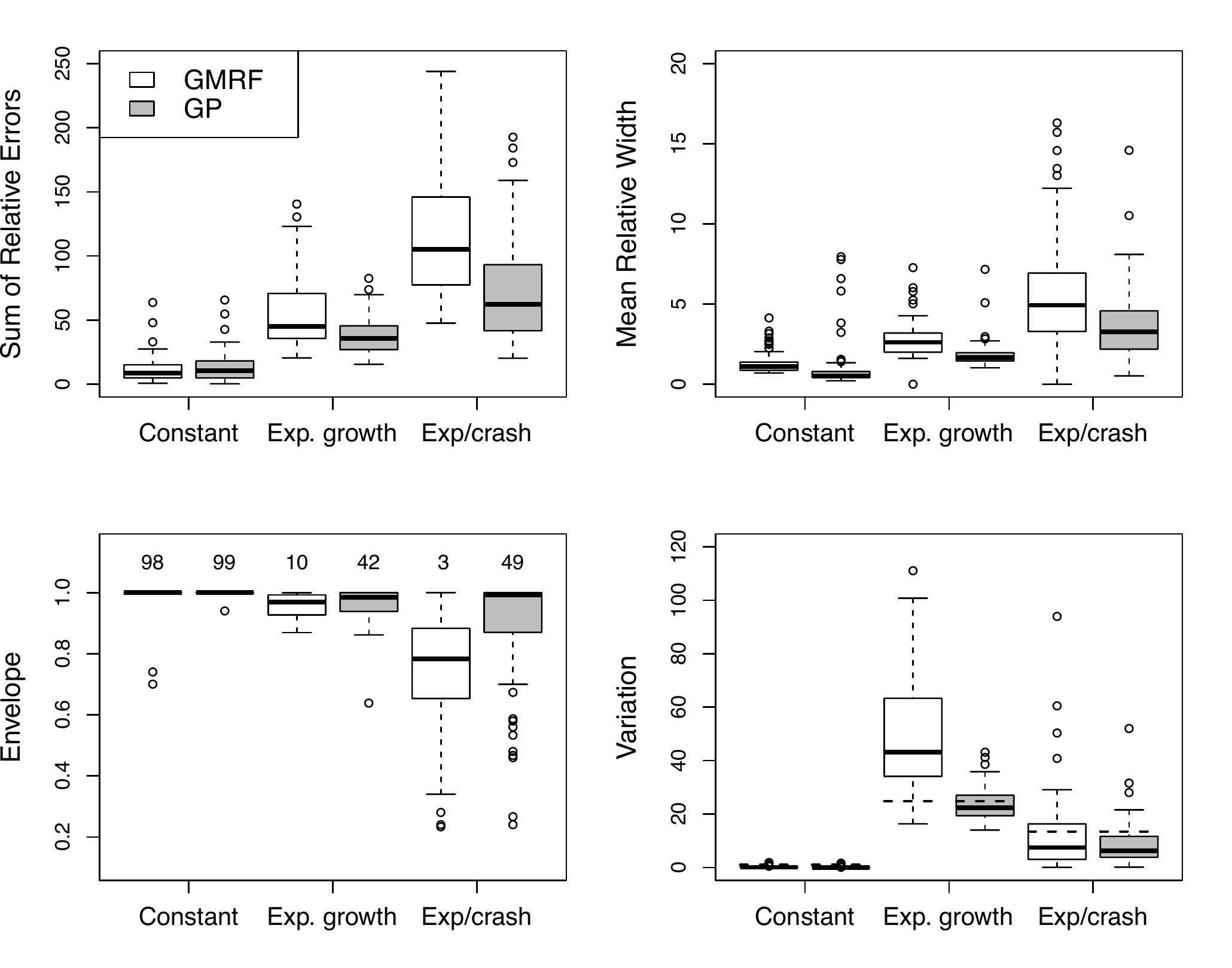}
  \end{center}
  \vspace{-10pt}
  \caption{\small{Boxplots of SRE (top left), MRW (top right), envelope (bottom left) and variation (bottom right)  based on 100 simulations for a constant trajectory, exponential growth and expansion followed by crash. The numbers above the boxplots of the bottom left plot represent the estimated frequentist coverage of the 95\% BCIs, and the dashed lines in the bottom right plot indicate variations of the true simulated trajectories.}}
\end{figure}
\subsection{Egyptian HCV}
Hepatitis C virus was first identified in 1989. By 1992, when HCV antibody testing became widely available,
the prevalence of HCV in Egypt was about $10.8\%$. Today, Egypt is the country with the highest HCV prevalence \citep{newHCV}.
A widely held hypothesis that can explain the epidemic emphasizes the role of \textbf{a} parenteral antischistosomal therapy (PAT) 
campaign, that started in the 1920s, combined with lack of sanitary practices. The campaign was discontinued in the 
1970s when the intravenous treatment was gradually replaced by oral administration of the treatment \citep{oldHCV}. Coalescent demographic
methods developed over the last 10 years demonstrated evidence in favor of this hypothesis \citep{Pybus01032003,drummond_bayesian_2005,minin_smooth_2008}. Therefore, this example is well suited for testing our method.
 We analyze the genealogy estimated by \citet{minin_smooth_2008}, based on $63$ HCV sequences sampled in 
Egypt in 1993, and compare our method to the GMRF smoothing method \citep{minin_smooth_2008}.
The results are depicted in Figure 4, with time scaled in units of years. In line with previous results \citep{Pybus01032003,drummond_bayesian_2005,minin_smooth_2008}, 
our estimation recovers the exponential growth of the HCV population size starting from the 1920s when the intravenously administered PAT was introduced. Both \citet{Pybus01032003} and \citet{minin_smooth_2008} hypothesize that the population trajectory
remained constant before the start of the exponential growth. The GMRF and GP approaches disagree the most on the HCV population size reconstruction prior to 1920s. The GP method produces narrower BCIs near the root of the genealogy (1710-1770) than the GMRF approach. In contrast, GP BCIs are inflated in the time period from 1770 to 1900 in comparison to the GMRF results. We believe that the uncertainty estimates produced by the GP approach are more reasonable than the GMRF result, because there are multiple coalescent events during 1710- 1770, providing information about the population size, while the time interval 1770 - 1900 has no coalescent events, a data pattern that should result in substantial uncertainty about the HCV population size. Another notable difference between the GMRF and GP methods is in estimation of the HCV population trajectory after 1970. The GP approach suggests a sharper decline in population size during this time interval. 
\begin{figure}
  \begin{center}
 \includegraphics[width=1.0\textwidth]{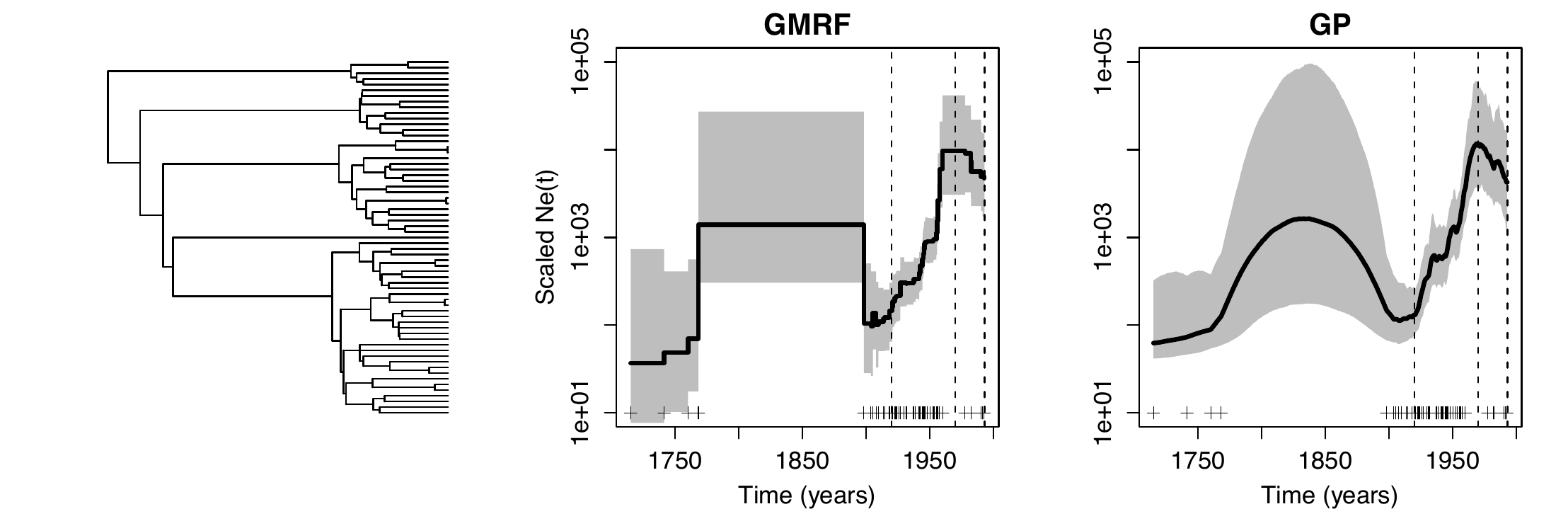}
  \end{center}
\caption{\small{Egyptian HCV. The first plot (left to right) is one possible genealogy reconstructed by \citet{minin_smooth_2008}. The next two plots represent the log of scaled effective population trajectory estimated using the GMRF smoothing method and our GP method. The posterior medians for the last two plots are represented by solid black lines and the $95\%$ BCI's are represented by the gray shaded areas. The vertical dashed lines mark the years 1920 (the start of intravenous PAT) , 1970 (the end of intravenous PAT) and 1993 (sampling time of sequences).}}
\end{figure}
\subsection{Seasonal Human Influenza}
Here, we estimate population dynamics of human influenza A, based on 288 H3N2 sequences sampled in New York state from January, 2001 to March, 2005. Sequences from the coding region of the influenza hemagglutinin (HA) gene of H3N2 influenza A virus from New York state were collected from the NCBI Influenza Database \citep{influenza_data}, incorporating the exact dates of viral sampling in weeks (heterochronous sampling) and aligned using the software package MUSCLE \citep{muscle}. These sequences form a subset of  sequences analyzed in \citep{rambaut_genomic_2008}. We carried out a phylogenetic analysis using the software package BEAST \citep{latest_beast} to generate a majority clade support genealogy with median node heights as our genealogical reconstruction. The reconstructed genealogy is depicted in the left plot of Figure~\ref{flu_figure}. Demography of H3N2 influenza A virus in temperate regions, such as New York, is characterized by epidemic peaks during winters followed by strong bottlenecks at the end of epidemic seasons. As expected, our method recovers the peaks in the effective number of infections during all seasons starting from the 2001-2002 flu season (flu seasons are represented as doted rectangles in Figure~\ref{flu_figure}). The GMRF method fails to recover the peak during the 2002-2003 season. The large uncertainty in population size estimation during the 1999-2000, 2000-2001, and at the beginning of 2005-2006 seasons is explained by the small number of coalescent events during those time periods, however, this uncertainty is larger in the GMRF recovered trajectory. During the 2001-2002 flu season, the GMRF method fails to recover the expected trajectory of a peak followed by a bottleneck and instead, this method recovers an epidemic that started during the end of 2001, increased and remained ``at peak" until the end of the following winter. The GMRF recovered trajectory during the winter season of 2003 exhibits a steep decrease. In contrast, the GP method detects a late peak during the 2001-2002 season, followed by a decline in the number of infections. There is a small bump in the effective population size of influenza in the winter of 2003, which is more realistic than a steady decline 
in the number of infections estimated by the GMRF method. Overall, we believe that the GP reconstructed trajectory is more feasible  from an epidemiological point of view than the GMRF population size reconstruction. 
\begin{figure}
  \begin{center}
 \includegraphics[width=1.0\textwidth]{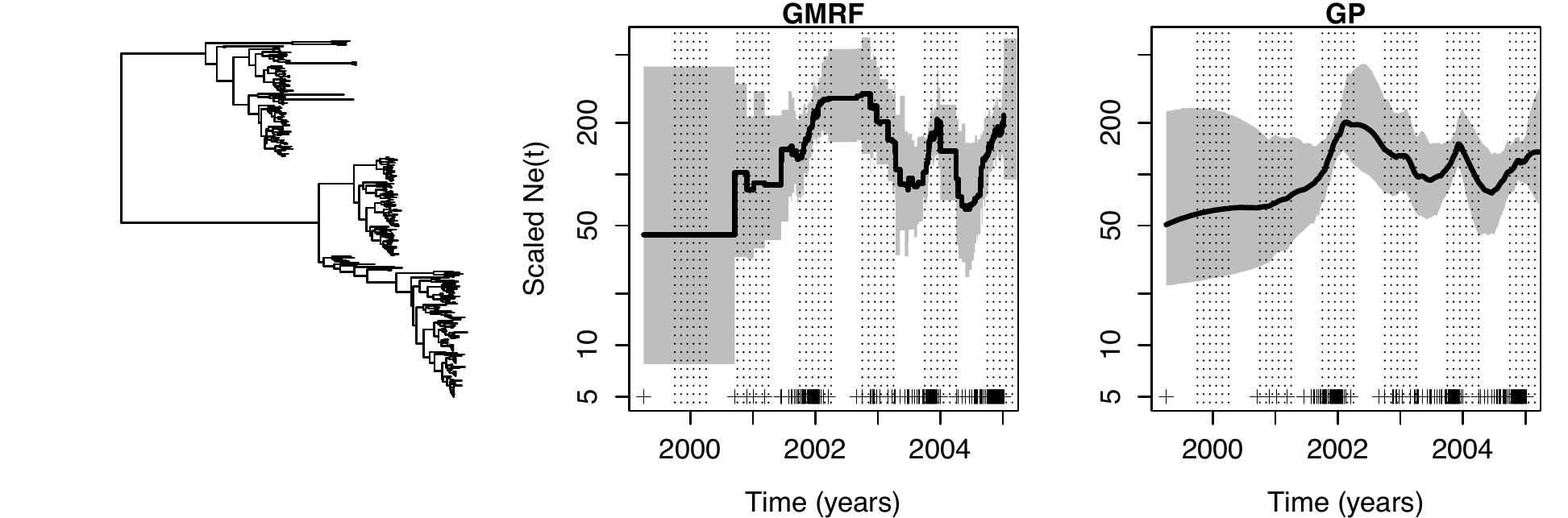}
  \end{center}
\caption{\small{H3N2 Influenza A virus in New York state. The first plot (left) is the estimated genealogy. The second and third plots are the GMRF and GP estimations of log scaled effective population trajectories. Winter seasons are represented by the doted shaded areas.  Posterior medians are represented by solid black lines and $95\%$ BCIs are represented by gray shaded areas.}}
\label{flu_figure}
\end{figure}

\section{Prior Sensitivity}
In all our examples, we placed a Gamma prior on the precision parameter $\theta$ with parameters $\alpha=0.001$ and $\beta=0.001$. This precision parameter, unknown to us \textit{a priori}, controls the smoothness of the GP prior. We investigate the sensitivity of our results to the Gamma prior specification using the Egyptian HCV data. In the first plot of Figure~\ref{sensitivity}, we show the prior and posterior distributions of $\theta$ under our default prior. The difference in densities suggests that prior choices do not have an impact on the posterior distribution. Since the mean of a Gamma distributed random variable is $\alpha/\beta$, we investigate the sensitivity by fixing $\beta=.001$ and setting the value of $\alpha$ to 0.001, 0.002, 0.005, 0.01 and 0.1, corresponding to prior means 1, 2, 5, 10 and 100 and variances 1000, 2000, 5000, 10000 and 100000, and by trying two extremes: $\alpha=1$,  $\beta=.0001$ and $\alpha=.001$, $\beta=1$, to examine the posterior distribution of $\theta$ under these priors. The posterior sample boxplots displayed in Figure 6 demonstrate that our results are fairly robust to different choices of $\alpha$. 
\begin{figure}[b]
  \begin{center}
 \includegraphics[width=\textwidth]{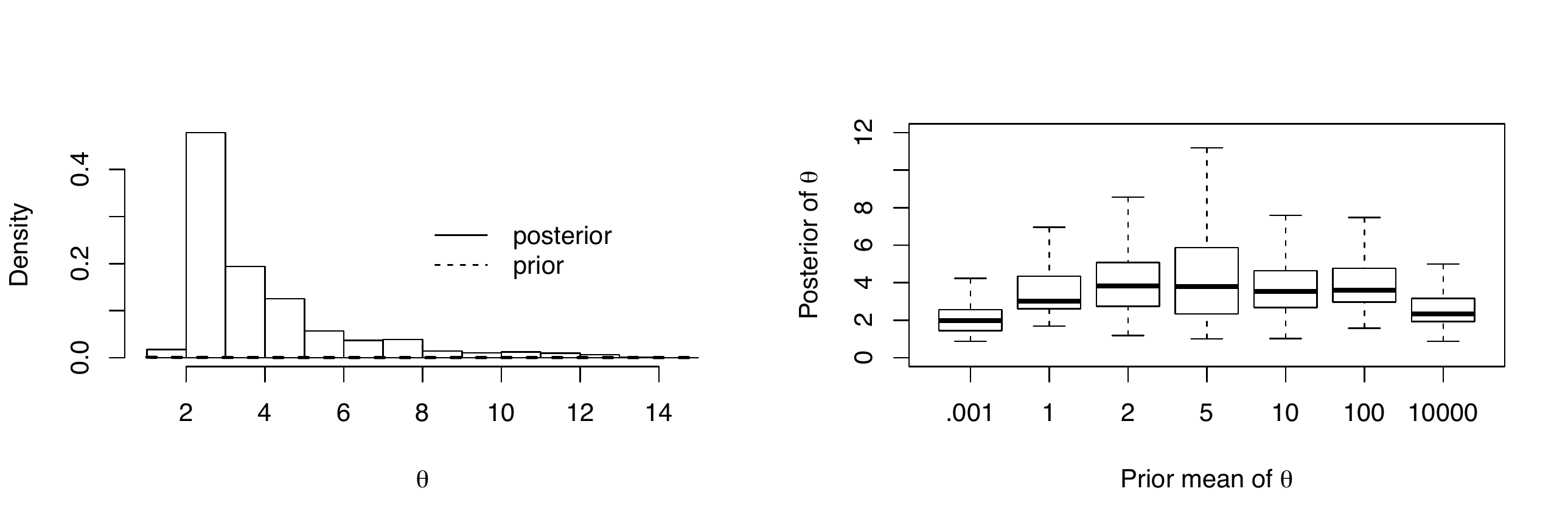}
  \end{center}
\caption{\small{Prior sensitivity on the GP precision parameter. Left plot shows the prior and posterior distributions represented by dashed line and vertical bars respectively. Right plot shows the boxplots of the posterior distributions of the precision parameter when the prior distributions differ in mean and variance of the precision parameter $\theta$. These plots are based on the Egyptian HCV data.}}
\label{sensitivity}
\end{figure}

\section{Sensitivity to the Order of the Gaussian Process}

We evaluate our GP-based method for three different Gaussian Process priors for the Egyptian HCV genealogy. In Figure 7, we show the recovered trajectories for Brownian Motion (BM), Ornstein-Uhlenbeck (OU) and approximated Integrated Brownian motion (IBM) \citep{Rw2}. The common characteristic of these three priors is the sparsity of their precision matrices (inverse covariance matrix), allowing for computational tractability. Figure 7 shows that the order of the process does make a difference, but only in regions with large posterior uncertainty, where prior influence is more pronounced.

\begin{figure}
\vspace{-25pt}
  \begin{center}
 \includegraphics[width=\textwidth]{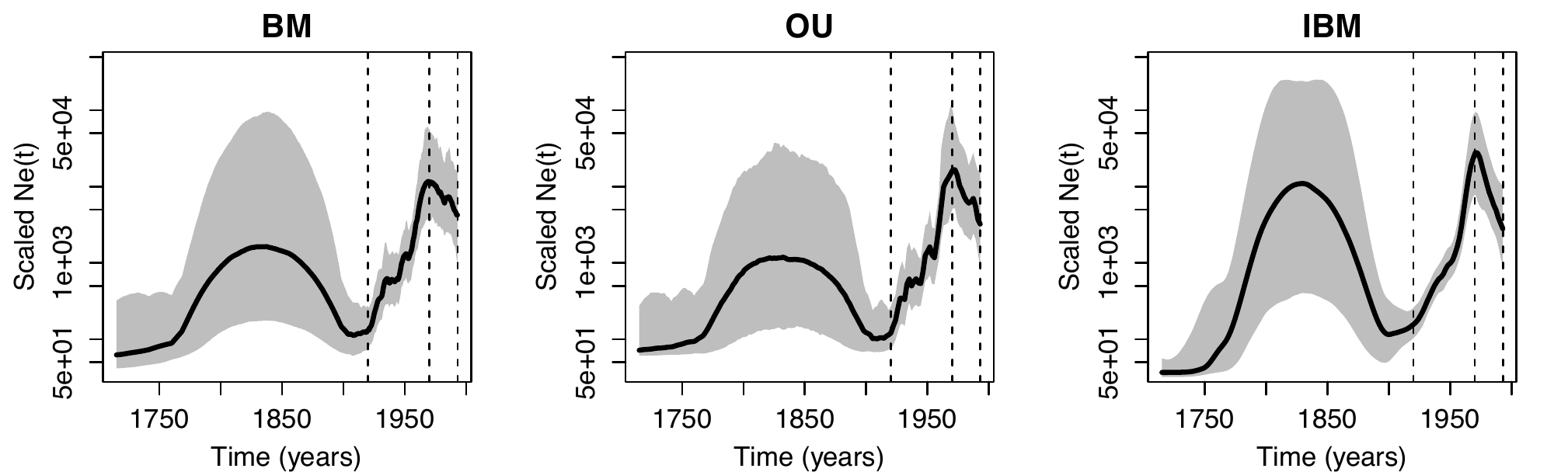}
  \end{center}
  \vspace{-10pt}
  \caption{\small{Egyptian HCV recoverd by placing three different Gaussian process priors. The first plot (left to right) corresponds to a Brownian motion (BM), the second -- to Ornstein-Uhlenbeck (OU) and the last one  -- to the approximated integrated  Brownian motion (IBM).}}
\end{figure}

\section{Discussion}
We propose a new nonparametric method for estimating population size dynamics from gene genealogies. To the best of our knowledge, we are the first to solve this inferential problem using modern tools from Bayesian nonparametrics. In our approach, we assume that the population size trajectory \textit{a priori} follows a transformed Gaussian process. This flexible prior allows us to model population size trajectory as a continuous function without specifying its parametric form and without resorting to artificial discretization methods. We tested our method on simulated and real data and compared it with the competing GMRF method. On simulated data, our method recovers the truth with better accuracy and precision. On real data, where true population trajectories are unknown, our method recovers known epidemiological aspects of the population dynamics and produces more reasonable estimates of uncertainty than the competing GMRF method.
\par
We bring Bayesian nonparametrics into the coalescent framework by viewing the coalescent as a point process. This representation allows us to adapt Bayesian nonparametric methods originally developed for Poisson processes to the coalescent modeling. In particular, it allows us to adapt the thinning-based data augmentation for Poisson processes developed by \citet{adams_tractable_2009}. We devise an analogous data augmentation for the coalescent by developing a series of new thinning algorithms for the coalescent. Although we use these algorithms in a very narrow context, our novel coalescent simulation protocols should be of interest to a wide range of users of the coalescent. For example, we are not aware of any competitors of our Web Supplementary Algorithms 2 and 4 that allow one to simulate coalescent times with a continuously and \textit{stochastically} varying population size.
\par
Our method works with any Gaussian process with mean $\mathbf{0}$ and covariance matrix $\mathbf{C}$, where the latter controls the level of smoothness and autocorrelation. For computational tractability however, sparsity in the precision matrix (inverse covariance matrix) may be necessary for complex trajectories with a high number of latent points. One way to achieve sparse matrix computations and computational tractability is to use GP that is also Markov. In all our examples, we use Brownian motion with precision parameter $\theta$; however, the nondifferentiability characteristic of the Brownian motion is compensated by the fact that our estimate of effective population trajectory is the posterior median evaluated pointwise, which is smoother than any of the sampled posterior curves. Additionally, we compared Brownian motion, Ornstein-Uhlenbeck and a higher order integrated Brownian motion for one of our examples and obtained very similar results under all three priors (see Web Supplementary Materials). Finite sample distributions under these three priors enjoy sparse precision matrices that yield computational tractability comparable to the GMRF method. In our Brownian motion prior, the precision parameter controls the level of smoothness of the estimated population size trajectory. We find that this important parameter shows little sensitivity to prior perturbations.
\par
Our method assumes that a genealogy or tree is given to the researcher. However, genealogies are themselves inferred from molecular sequences, so we need to incorporate genealogical uncertainty into our estimation. Our framework can be extended to inference from molecular sequences instead of genealogies by introducing another level of hierarchical modeling into our Bayesian framework, similar to the work of \citet{drummond_bayesian_2005} and \citet{minin_smooth_2008}. Further, we plan to extend our method to handle molecular sequence data from multiple loci as in \citep{heled_bayesian_2008}. Finally, we would like to extend our nonparametric estimation into a multivariate setting, so that we can estimate cross correlations between population size trajectories and external time series. Estimating such correlations is a critical problem in molecular epidemiology.
\par
We deliberately adapted the work of \citet{adams_tractable_2009} on estimating the intensity function of an inhomogeneous Poisson process, as opposed to alternative ways to attack this estimation problem \citep{mller_log_1998, kottas_bayesian_2005}, to the coalescent. We believe that among the state-of-the-art Bayesian nonparametric methods, our adopted GP-based framework is the most suitable for developing the aforementioned extensions. First, it is straightforward to incorporate external time series data into our method by replacing a univariate GP prior with a multivariate process that evolves the population size trajectory and another variable of interest in a correlated fashion \citep{Teh05semiparametriclatent}. Second, the fact that our method does not operate on a fixed grid is critical for relaxing the assumption of a fixed genealogy, because fixing the grid \textit{a priori} is problematic when one starts sampling genealogies, including coalescent times, during MCMC iterations.  
\par
Finally, since the coalescent model with varying population size can be viewed as a particular example of an inhomogeneous continuous-time Markov chain, all our mathematical and computational developments are almost directly transferable to this larger class of models. Therefore, our developments potentially have implications for nonparametric estimation of inhomogeneous continuous-time Markov chains with numerous applications.

\section*{Acknowledgements}
We thank Peter Guttorp, Joe Felsenstein, and Elizabeth Thompson for helpful discussions. VNM was supported by the NSF grant No.\ DMS-0856099. JAP acknowledges scholarship from CONACyT Mexico to pursue her doctoral work.

\renewcommand{\baselinestretch}{1.8}

\bibliographystyle{apalike}
\bibliography{skytrack.bib}
\pagebreak
\appendix
\setcounter{section}{0}

\section{Coalescent Simulation Algorithms}
\renewcommand\thesection{\Alph{section}}
\setcounter{equation}{0}%
\numberwithin{equation}{section}
\begin{prop}
Algorithm 1 generates $t_{n}<t_{n-1}<\cdots<t_{1}$, such that
\begin{equation}
P(t_{k-1}>t|t_{k})= \exp\left[ -\int^{t}_{t_{k}}\frac{C_{k}dx}{N_{e}(x)}\right],
\end{equation}
where $N_{e}(t)$ is known deterministically.
\end{prop}
\begin{proof}
Let $T_{i}=t_{k}+E_{1}+\ldots+E_{i}$, where $\{E_{i}\}^{\infty}_{i=1}$ are iid exponential  $Exp(C_{k}\lambda)$ random numbers. Given $t_{k}$, Algorithm 1 generates and accumulates iid exponential random numbers until $T_{i}$ is accepted with probability $1/\lambda N_{e}(T_{i})$, in which case, $T_{i}$ is labeled $t_{k-1}$. Let $N(t_{k},t]=\#\{i \geq 1: t_{k}<T_{i}\leq t\}$ denote the number of iid exponential random numbers generated in $(t_{k},t]$. Then, $\{N(t_{k},t],t>t_{k}\}$ constitutes a Poisson process with intensity $C_{k}\lambda$. Then, given  $N(t_{k},t]=1$, the conditional density of a point $x$ in $(t_{k},t]$ is $1/(t-t_{k})$ and the probability of accepting such a point as a coalescent time point with variable population size is $1/\lambda N_{e}(x)$. Hence
\begin{equation}
P(t_{k-1} \leq t | t_{k},N(t_{k},t]=1)=\frac{1}{\lambda(t-t_{k})}\int^{t}_{t_{k}}\frac{dx}{N_{e}(x)},
\end{equation}
and
\begin{equation}
P(t_{k-1} > t | t_{k},N(t_{k},t]=m)=\left( 1-\frac{1}{\lambda(t-t_{k})}\int^{t}_{t_{k}}\frac{dx}{N_{e}(x)}\right) ^{m}.
\end{equation}
Then,
\begin{equation*}
P(t_{k-1}>t|t_{k})=\sum^{\infty}_{m=1}P(t_{k-1}>t|t_{k},N(t_{k},t]=m)P(N(t_{k},t]=m)
\end{equation*}
\begin{equation*}
=\sum^{\infty}_{m=1}\left( 1-\frac{1}{\lambda(t-t_{k})}\int^{t}_{t_{k}}\frac{dx}{N_{e}(x)}\right) ^{m}\frac{\left( C_{k}\lambda (t-t_{k})\right)^{m} \exp\left[ -C_{k} \lambda (t-t_{k}) \right] }{m!}
\end{equation*}
\begin{equation*}
= \exp\left[ -C_{k} \lambda (t-t_{k}) \right]\sum^{\infty}_{m=1}\frac{\left( C_{k}\lambda(t-t_{k})-C_{k}\int^{t}_{t_{k}}\frac{dt_{k-1}}{N_{e}(t_{k-1})}\right) ^{m}}{m!}
\end{equation*}
\begin{equation*}
=\exp\left[ -\int^{t}_{t_{k}}\frac{C_{k}dx}{N_{e}(x)}\right].
\end{equation*}
\end{proof}

\renewcommand{\algorithmicrequire}{\textbf{Input:}} 
\renewcommand{\algorithmicensure}{\textbf{Output:}}
\setcounter{algorithm}{1}
\linespread{0.9}
\begin{algorithm}
\caption{Simulation of isochronous coalescent times by thinning with $f(t) \sim \mathcal{GP}(\mathbf{0},\mathbf{C}(\theta))$} \label{alg3} 
\begin{algorithmic} [1]
\REQUIRE {$k=n$, $t_{n}=0$, $t=0$, $i_{j}=0$, $m_{j}=0$, $j=2,\ldots,n$, $\lambda$}
\ENSURE {$\mathcal{T}=\{t_{k}\}^{n}_{k=1}$, $\mathcal{N}=\{\{t_{k,i}\}^{m_{k}}_{i=1}\}^{n}_{k=2}$, $\mathbf{f}_{\mathcal{T},\mathcal{N}}$}
\WHILE {$k>1$}
\STATE Sample $E\sim Exponential (C_{k}\lambda)$ and $U\sim U(0,1)$
\STATE t=t+E
\STATE Sample $f(t) \sim P(f(t)|\{f(t_{l})\}^{n}_{l=k},\{\{f(t_{l,i})\}^{m_{l}}_{i=1}\}^{n}_{l=k};\theta)$
\IF{$U \leq \frac{1}{1+\exp(-f(t))} $}
\STATE $k \leftarrow k-1$, $t_{k} \leftarrow t$
\ELSE
\STATE  $i_{k}\leftarrow i_{k}+1$, $m_{k}\leftarrow m_{k}+1$, $t_{k,i_{k}}\leftarrow t$
\ENDIF
\ENDWHILE
\end{algorithmic} 
\end{algorithm}

Algorithm 3 and 4 are analogous heterochronous versions of Algorithm 1 and 2.

\linespread{0.9}
\begin{algorithm}
\caption{Simulation of heterochronous coalescent by thinning - $N_{e}(t)$ is a deterministic function} \label{alg2} 
\begin{algorithmic} [1]
\REQUIRE {$n_{1},n_{2},\ldots,n_{m}$, $s_{1},\ldots,s_{m}$, $1/N_{e}(t) \leq \lambda$, $N_{e}(t)$, $m$}
\ENSURE {for $n=\sum^{m}_{j=1}n_{j}$, $\mathcal{T}=\{t_{k}\}^{n}_{k=1}$}
\STATE $i=1$, $j=n-1$, $n=n_{1}$, $t=t_{n}=s_{1}$
\WHILE {$i<m+1$}
\STATE Sample $E\sim Exp(\binom{n}{2}\lambda)$ and $U\sim U(0,1)$
\IF{$U \leq \frac{1}{N_{e}(t+E)\lambda} $}
\IF{$t+E<s_{i+1}$ }
\STATE $t_{j} \leftarrow t \leftarrow t+E$
\STATE $j \leftarrow j-1$, $n \leftarrow n-1$ 
\IF {$n>1$}
\STATE go to 2
\ELSE
\STATE go to 14
\ENDIF
\ELSE
\STATE $i \leftarrow i+1$, $ t \leftarrow s_{i}$, $n \leftarrow n+n_{i}$
\ENDIF
\ELSE 
\STATE $t \leftarrow t+E$
\ENDIF 
\ENDWHILE
\end{algorithmic} 
\end{algorithm}

\renewcommand{\algorithmicrequire}{\textbf{Input:}} 
\renewcommand{\algorithmicensure}{\textbf{Output:}}
\linespread{0.9}
\begin{algorithm}
\caption{Simulation of heterochronous coalescent by thinning with $f(t) \sim \mathcal{GP}(\mathbf{0},\mathbf{C}(\theta))$} \label{alg4} 
\begin{algorithmic} [1]
\REQUIRE {$n_{1},n_{2},\ldots,n_{m}$, $s_{1}=0,\ldots,s_{m}$, $i_{j}=0$, $m_{j}=0$, $j=2,\ldots,n$, $\lambda$, $m$}
\ENSURE {for $n=\sum^{m}_{j=1}n_{j}$, $\mathcal{T}=\{t_{k}\}^{n}_{k=1}$, $\mathcal{N}=\{\{t_{k,i}\}^{m_{k}}_{i=1}\}^{n}_{k=2}$, $\mathbf{f}_{\mathcal{T},\mathcal{N}}$}
\STATE $i=1$, $j=n-1$, $n=n_{1}$, $t=t_{n}=s_{1}$
\WHILE {$i<m+1$}
\STATE Sample $E\sim Exp(\binom{n}{2}\lambda)$ and $U\sim U(0,1)$
\STATE Sample $f(t+E) \sim P(f(t+E)|\{f(t_{l})\}^{n}_{l=k},\{\{f(t_{l,i})\}^{m_{l}}_{i=1}\}^{n}_{l=k};\theta)$
\IF{$U \leq \frac{1}{1+\exp(-f(t+E))}  $}
\IF{$t+E<s_{i+1}$ }
\STATE $t_{j} \leftarrow t \leftarrow t+E$
\STATE $j \leftarrow j-1$, $n \leftarrow n-1$ 
\IF {$n>1$}
\STATE go to 2
\ELSE
\STATE go to 14
\ENDIF
\ELSE
\STATE $i \leftarrow i+1$, $ t \leftarrow s_{i}$, $n \leftarrow n+n_{i}$
\ENDIF
\ELSE 
\IF {$t+E<s_{i+1}$}
\STATE $t_{j+1,i_{j+1}} \leftarrow t+E$, $i_{j+1}\leftarrow i_{j+1}+1$
\ENDIF
\STATE $t \leftarrow t+E$
\ENDIF 
\ENDWHILE
\end{algorithmic} 
\end{algorithm}

An \texttt{R} implementation of these algorithms is available at \\
\url{ http://www.stat.washington.edu/people/jpalacio.}

\clearpage

\section{MCMC Sampling}
\renewcommand\thesection{\Alph{section}}%
\setcounter{equation}{0}%
\numberwithin{equation}{section}%

\renewcommand\thesection{\setcounter{section}{2}\Alph{section}}
Since the coalescent under isochronous sampling is a particular case of the coalescent model under heterochronous sampling, we employ here the notation of the heterochronous coalescent, understanding that $C_{0,k}=C_{k}$, $I_{0,k}=(t_{k},t_{k-1}]$ and $i=0$ for isochronous data.\\
{\bf{Sampling the number of latent points}}. A reversible jump algorithm is constructed for the number of ``rejected'' points. We propose to add or 
remove points with equal probability in each interval defined by Equations (3) and (4). 
When adding a point in a particular interval, we propose a location uniformly from the interval and its 
predicted function value $f(t^{*}) \sim P(f(t^{*})|\mathbf{f}_{\mathcal{T},\mathcal{N}},\theta)$. 
When removing a point, we propose to remove a point selected uniformly from the pool of rejected points in that particular interval. We add points with proposal distributions $q^{i,k}_{up}$ and remove points with proposal distributions $q^{i,k}_{down}$. Then,
\begin{equation}
q^{i,k}_{up}=\frac{P(f(t^{*})|\mathcal{T},\mathcal{N},\theta)}{2l(I_{i,k})},
\end{equation}
\begin{equation}
q^{i,k}_{down}=\frac{1}{2m_{i,k}},
\end{equation}
and the acceptance probabilities are:
\begin{equation}
a^{i,k}_{up}=\frac{ l(I_{i,k})\lambda C_{i,k}}{(m_{i,k}+1)(1+e^{f(t^{*})})},
\end{equation}
\begin{equation}
a^{i,k}_{down}=\frac{m_{i,k}(1+e^{f(t^{*})})}{ l( I_{i,k})\lambda C_{i,k}}.
\end{equation}
\\
{\bf{Sampling locations of latent points}}. We use a Metropolis-Hastings algorithm to update the locations of 
latent points. We first choose an interval defined by Equations (3) and (4) with probability proportional to its length and we then propose point locations uniformly at random in that interval together with their predictive function values $\mathbf{f}_{\mathbf{t}^{*}} \sim P(\mathbf{f}_{\mathbf{t}^{*}}|\mathbf{f}_{\mathcal{T},\mathcal{N}},\theta)$. Since the proposal distributions are symmetric, the acceptance probabilities are:
\begin{equation}
a^{i,k}=\frac{1+e^{f(t)}}{1+e^{f(t^{*})}}.
\end{equation}
\\
{\bf{Sampling transformed effective population size values}}. We use an elliptical slice sampling proposal described in \citep{murray_elliptical_2010}. In both cases, isochronous or heterochronous, the full conditional distribution of the function values $\mathbf{f}_{\mathcal{T},\mathcal{N}}$ is proportional to the product of a Gaussian density and the thinning acceptance and rejection probabilities:
\begin{equation}
P(\mathbf{f}_{\mathcal{T},\mathcal{N}}|\mathcal{T},\mathcal{N},\lambda,\theta) \propto P(\mathbf{f}_{\mathcal{T},\mathcal{N}}|\theta) L(\mathbf{f}_{\mathcal{T},\mathcal{N}}),
\end{equation}
where
 \begin{equation}
 L(\mathbf{f}_{\mathcal{T},\mathcal{N}})=\prod^{n}_{k=2}\left( \frac{1}{1+e^{-f(t_{k-1})}} \right)
\prod^{m_{k}}_{i=1}\frac{1}{1+e^{f(t_{k,i})}}.
\end{equation}
{\bf{Sampling hyperparameters}}. 
The full conditional of the precision parameter $\theta$ is a Gamma distribution. Therefore, we update $\theta$ by drawing from its full  conditional:
\begin{equation}
\theta|\mathbf{f}_{\mathcal{T},\mathcal{N}},\mathcal{T},\mathcal{N} \sim Gamma \left( \alpha^{*}=\alpha+\frac{\#\{\mathcal{N} \cup \mathcal{T}\}}{2},
\beta^{*}=\beta + \frac{\mathbf{f}_{\mathcal{T},\mathcal{N}}^{t}Q\mathbf{f}_{\mathcal{T},\mathcal{N}}}{2}\right),
\end{equation}
where $Q=\frac{1}{\theta}C^{-1}$.\\ 
For the upper bound $\lambda$ on $N_{e}(t)^{-1}$, we use the Metropolis-Hastings update by proposing new values using a uniform  proposal reflected at $0$. That is, we propose $\lambda^{*}$ from $U(\lambda-a,\lambda+a)$. If the proposed value $\lambda^{*}$ is negative, we flip its sign. Since the proposal distribution is symmetric, the acceptance probability is:
\begin{equation}
a=\frac{P(\lambda^*)}{P(\lambda)}\left(\frac{\lambda^*}{\lambda} \right)^{\#\left\lbrace \mathcal{N} \cup \mathcal{T} \right\rbrace} \exp\left[ -\left( \lambda^*-\lambda\right) \sum^{n}_{k=2}\sum^{m_{k}}_{i=1}C_{i,k}l(I_{i,k})\right],
\end{equation} 
where $P(\lambda)$ is defined in Equation (7).

\renewcommand{\baselinestretch}{1.8}

\end{document}